\documentclass[journal]{IEEEtran}

\usepackage{amsmath,amssymb,amsfonts}
\usepackage{amsthm}
\usepackage{graphicx}
\usepackage{booktabs}
\usepackage{array}
\usepackage{cite}
\usepackage{url}
\usepackage{balance}
\usepackage{tikz}
\usetikzlibrary{arrows.meta,positioning,calc,decorations.pathmorphing,patterns}

\graphicspath{{figures/}}

\tikzset{
  tip/.style={circle,draw,fill=black,minimum size=3.4pt,inner sep=0pt},
  sig/.style={-{Latex[length=2mm]},line width=0.6pt},
  dist/.style={-{Latex[length=2mm]},red!75,line width=0.9pt}
}

\newtheorem{theorem}{Theorem}

\newtheorem{proposition}{Proposition}
\newtheorem{remark}{Remark}

\begin{document}

\title{Interaction Dynamics Modeling and Force-Limited Control for Steerable Catheter--Tissue Interaction}

\author{Yongyan~Cao%
\thanks{Preprint, June 2026. This work has been submitted for possible publication; copyright may be transferred without notice.}%
\thanks{Y.~Cao is with Voryx Robotics LLC, San Jose, CA 95136, USA. Corresponding author: Yongyan Cao (e-mail: yongyancao@gmail.com).}%
}

\markboth{Preprint}%
{Cao: Interaction Dynamics Modeling and Force-Limited Control for Steerable Catheter--Tissue Interaction}

\maketitle

\begin{abstract}
Steerable catheter control couples planned tip motion, compliance against moving tissue, friction rejection, and force limitation. We formulate catheter--tissue interaction dynamics in the scalar tip-normal coordinate of a single-segment single-tendon catheter. Partial-physics feedforward cancels the characterized nominal bending dynamics, exposing a linear interaction model whose input gain varies with the scalar catheter inertia. From this model we derive a finite-horizon quadratic program with disturbance centering and corrective-force constraints. The current MuJoCo validation evaluates its static impedance specialization with integral residual compensation, tendon saturation, and a pointwise corrective-force limit; a separate reduced-model OSQP test verifies the corresponding horizon box constraint. On an eight-link tendon-driven catheter, integral compensation reduces free-space approach error from 0.28 to 0.03\,mm. Against a penetrating target, the unrestricted controller reaches 0.60\,N, whereas the force-limited variant preserves the same approach accuracy with a measured 0.47\,N peak under a configured 0.5\,N engineering threshold. The measured peak remains 0.47\,N under the tested $0.5$\,mm, $1.2$\,Hz wall motion. These simulation results isolate the interaction between offset-free motion regulation and force limiting and define the next step toward an online horizon implementation and integrated-platform validation.
\end{abstract}


\section{Introduction}

\subsection{Clinical Motivation}

\IEEEPARstart{S}{teerable} catheters are the primary tool for cardiac electrophysiology (EP) procedures including radiofrequency ablation, where the tip must be positioned precisely at target tissue while maintaining controlled, stable contact. The central control problem is therefore not merely tip tracking and not merely force regulation; it is the regulation of \emph{catheter--tissue interaction dynamics}. The interaction state must encode how the tip moves relative to tissue, how persistent friction and contact forces bias that motion, and how safety limits reshape what motion is physically allowable.

Existing methods regulate these interaction dynamics through different mechanisms. Classical impedance control~\cite{hogan1985} shapes the tip port as a virtual mechanical impedance $Z(s) = M_d s^2 + D_d s + K_d$, providing passive compliance without an explicit contact model. Three complementary design requirements motivate the present formulation: \textbf{(i)}~an explicit force-related constraint, \textbf{(ii)}~compensation for steady error under persistent loading, and \textbf{(iii)}~a prediction model that can incorporate trajectory and actuator information.

\subsection{Challenges Specific to Catheters}

Catheter control emphasizes mechanics and sensing effects that are less prominent in rigid manipulators. \emph{Model variability:} distributed elasticity, sheath friction, tendon backlash, and hysteresis vary across specimens, favoring reduced-order feedback over complete model inversion~\cite{antman1995,rucker2011}. \emph{Cardiac motion:} gross target tracking leaves a local quasi-periodic wall-normal component with uncertain phase and amplitude. \emph{Sensing:} tip position may be available through electromagnetic tracking, fluoroscopy, intracardiac echocardiography, or shape sensing, while direct force sensing is platform dependent. \emph{Force management:} contact force is associated with lesion formation and adverse events~\cite{yokoyama2008}; this study therefore uses 0.5\,N as a configured engineering threshold for simulation rather than as a universal tissue-damage boundary.

\subsection{Related Work}

\emph{Catheter and continuum-robot control} has progressed from PID and Jacobian-based kinematic regulation~\cite{camarillo2008,yip2014} to model-less feedback~\cite{yip2014} and predictive interaction-dynamics optimization. Cosserat-based predictive control can model rich catheter dynamics but requires online nonlinear-program solution, limiting update rates to ${\sim}10$--20\,Hz. Constant-curvature kinematics~\cite{webster2010} underlie most reduced-order interaction models. Learning-based approaches improve robustness to model uncertainty but usually do not provide formal stability or hard interaction-constraint guarantees.

\emph{Sensor-free contact-force control} is a directly relevant line. Jolaei \emph{et al.}~\cite{jolaei2020} regulate the contact force of a tendon-driven ablation catheter \emph{without a force sensor} via position control plus a displacement-based viscoelastic contact model (0.03--0.05\,N RMS), and Kesner and Howe~\cite{kesner2014} combine ultrasound guidance with force control on a moving cardiac target (${\sim}0.08$\,N RMS). These works regulate one component of the interaction dynamics---force---to a setpoint. The present framework instead treats force as a hard never-exceed interaction-dynamics constraint while regulating the tip-motion state, a complementary objective (Section~\ref{sec:compare}).

\emph{Virtual-impedance interaction shaping} has been demonstrated for catheter tip-force regulation with empirically tuned parameters. In the terminology of this paper, impedance is one way to realize interaction dynamics, but the configuration-dependent effective inertia and compliance are typically not addressed, giving inconsistent behavior across the workspace.

\emph{Disturbance rejection in flexible robots} can also be interpreted as interaction-dynamics regulation. Active disturbance rejection control (ADRC) and its extended state observer~\cite{han2009} lump unknown dynamics into an estimated disturbance state. The connection to offset-free predictive control~\cite{pannocchia2003,maeder2009} provides a formal framework for combining disturbance estimation with constraint-aware interaction optimization~\cite{rawlings2017}.

\emph{Base framework.} This work builds on the interaction-dynamics framework introduced for redundant-manipulator physical human--robot interaction~\cite{cao2026phri}: nonlinear robot dynamics are transformed into a configuration-invariant linear interaction model, then regulated by predictive optimization with disturbance augmentation and safety constraints. We carry that hierarchy---Interaction Dynamics $\rightarrow$ Configuration-Invariant Dynamics $\rightarrow$ Predictive Optimization---to the steerable-catheter setting, addressing the catheter-specific challenges of single-DOF actuation, severe model uncertainty, and hard contact-force safety.

The contributions of this paper are:

\textbf{Correct single-DOF formulation.} A single-tendon catheter has one controllable DOF (the curvature); we control the scalar 
tip-normal coordinate, avoiding the common over-parameterization that treats the 2-D tip as independently actuable.

\textbf{Configuration-invariant interaction model.} Partial-physics feedforward cancels characterized nominal stiffness and damping, yielding a constant-$A_d$, parameter-varying-$B_d$ interaction model. Remark~\ref{rmk:compress} interprets the integrating disturbance state as the observer-accessible low-frequency component of the aggregate residual.

\textbf{Predictive interaction-dynamics formulation.} We derive a finite-horizon QP and establish (Theorem~\ref{thm:equiv}) that its unconstrained, disturbance-free realization recovers a classical catheter impedance law. A separate reduced-model OSQP check verifies the corrective-force box constraint.

\textbf{Force-limited distributed-compliance evaluation.} The MuJoCo benchmark maps a pointwise corrective-force limit through a calibrated tendon-to-tip transmission and applies the physical tendon saturation. Measured tissue-force outcomes are reported separately from the constrained proxy.

\textbf{Mechanism characterization.} Simulation on a distributed-compliance physics plant quantifies free-space residual rejection, the contact-force effect of the corrective-force limit, cardiac-wall sensitivity, contact-normal misalignment, approach speed, and tendon hysteresis.

\section{Interaction Dynamics Formulation}\label{sec:model}

\begin{figure}[t]
\centering
\resizebox{\columnwidth}{!}{%
\begin{tikzpicture}[scale=1.0,every node/.style={font=\footnotesize}]
  \node[draw,fill=gray!15,minimum width=9mm,minimum height=8mm] (act) at (-2.85,0) {};
  \node[align=center,font=\scriptsize] at (-2.85,0) {tendon\\actuator};
  \draw[sig] (-2.3,0.32)--(-1.6,0.32);
  \node[above,font=\scriptsize] at (-1.95,0.34) {disp.\ $u$};
  \draw[line width=3pt,gray!70,line cap=round] (-2.4,0)--(0,0);
  \draw[line width=3pt,gray!70,line cap=round] (0,0) arc (270:360:1.8);
  \node[below,font=\scriptsize] at (-1.2,-0.12) {proximal sheath};
  \draw[line width=0.9pt] (0,0.2) arc (270:360:1.6);
  \node[font=\scriptsize,align=left] at (0.95,0.55) {tendon\\(force, $u$)};
  \draw[sig,gray] (2.55,0.55)--(1.45,0.62);
  \node[right,font=\scriptsize,align=left] at (2.5,0.55) {bending segment\\ $L,\;\kappa$; \ $M,C,K$};
  \node[font=\scriptsize] at (0,1.8) {$+$};
  \draw[dashed] (0,1.8)--(1.8,1.8);
  \draw[dashed] (0,1.8)--(0,0);
  \node[above,font=\scriptsize] at (0.75,1.84) {$R=1/\kappa$};
  \draw[-{Latex[length=1.3mm]}] (0,1.0) arc (270:355:0.8);
  \node[font=\scriptsize] at (0.42,1.28) {$\kappa L$};
  \node[tip] (T) at (1.8,1.8){};
  \node[right,font=\scriptsize,align=left] at (2.02,1.55) {tip $p$,\\ $y=n^\top p$};
  \draw[sig] (1.5,1.9)--(1.5,2.42);
  \node[right,font=\scriptsize] at (1.54,2.27) {$n$};
  \fill[pattern=north east lines,pattern color=gray] (0.6,2.45) rectangle (3.0,2.67);
  \draw[line width=0.8pt] (0.6,2.45)--(3.0,2.45);
  \node[above,font=\scriptsize] at (1.8,2.69) {tissue wall ($k_t,b_t$)};
  \draw[dist] (2.2,2.43)--(2.2,1.9);
  \node[red!75,right,font=\scriptsize] at (2.25,2.2) {$F_n\le F_\text{safe}$};
  \node[font=\scriptsize,align=center] at (-0.6,-0.62)
        {friction, backlash, hysteresis, contact $\Rightarrow d_\text{cat}$};
\end{tikzpicture}}
\caption{Single-segment tendon-actuated catheter plant, mapping the variables
of this section onto the physical parts: tendon displacement $u$ (control
input) drives the curvature $\kappa$ of a length-$L$ bending segment with
effective bending inertia $M$, damping $C$, and hysteretic stiffness $K$
\eqref{eq:plant}; the constant-curvature arc has radius $R=1/\kappa$ and
subtends $\kappa L$. The tip position $p$ and its controlled normal coordinate
$y=n^\top p$ press on a Kelvin--Voigt tissue wall ($k_t,b_t$), producing the
normal contact force $F_n$ held under the hard safety bound $F_\text{safe}$.
Sheath friction, tendon backlash, hysteresis, and contact are lumped into the
disturbance $d_\text{cat}$, compressed by the augmented Kalman estimator into
the observed state $d$ (Remark~\ref{rmk:compress}).}
\label{fig:plant}
\end{figure}
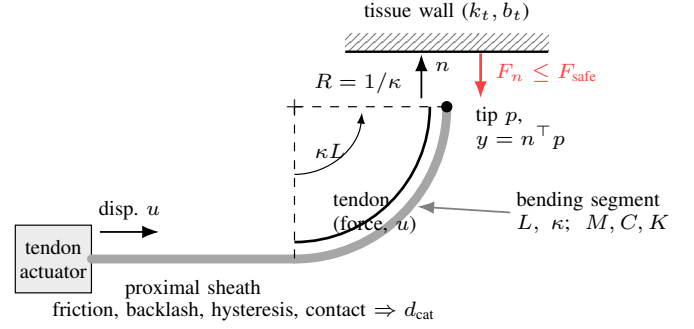

\subsection{Catheter Mechanics and Degrees of Freedom}

Consider a planar single-segment tendon-actuated catheter; tip position $p \in \mathbb{R}^2$ is set by tendon displacement $u \in \mathbb{R}$ (positive pull gives curvature $\kappa > 0$). The dominant bending dynamics take the rigid-body-analogous form
\begin{equation}
M(\kappa)\ddot{\kappa} + C(\kappa,\dot\kappa)\dot{\kappa} + K(\kappa) = u + d_{\text{cat}},
\label{eq:plant}
\end{equation}
with effective bending inertia $M$, damping $C$, nonlinear (hysteretic) stiffness $K$, and a lumped term $d_{\text{cat}}$ collecting tissue contact, friction, backlash, and hysteresis. This equation is not used as a full Cosserat model; it is the starting point for constructing a low-dimensional interaction-dynamics state. The correspondence $\kappa\leftrightarrow q$, $K\leftrightarrow G$ bridges to the configuration-invariant interaction-dynamics framework established for redundant manipulators in physical human--robot interaction~\cite{cao2026phri}; the present paper specializes that hierarchy to the single-DOF catheter setting.

\emph{Degrees of freedom.} With a single tendon, the only controllable coordinate is $\kappa$; the tip pose $p(\kappa)$ traces a one-parameter curve, so the two tip coordinates are \emph{not} independently actuable---only motion along the Jacobian direction $J_\kappa$ is. We therefore control a scalar coordinate (the curvature, or equivalently the tip displacement $y$ along the contact normal) and recover the 2-D pose through the kinematics of Section~\ref{sec:kin}.

\begin{remark}[Disturbance-compression interpretation]\label{rmk:compress}
The true residual after partial-physics cancellation, $\Delta(t)=f_{\text{real}}-f_{\text{nominal}}$, superposes effects on disparate timescales---quasi-static sheath friction and hysteresis, fast tissue contact, and quasi-periodic cardiac loading. The integrating model targets the low-frequency, observer-accessible component
\begin{equation}
d \;=\; \mathcal{P}_{\mathcal D}(\Delta),
\label{eq:compress}
\end{equation}
where $\mathcal D$ denotes the bandwidth represented by the estimator and integrating internal model. This is a reduced-order modeling interpretation: constant matched components admit offset-free cancellation in the unconstrained nominal model, while contact-onset transients and periodic harmonics are evaluated through force limiting and stress tests. The current MuJoCo benchmark realizes this idea with an integral force-bias state; the Kalman realization in \eqref{eq:aug} remains the estimator design for a future online horizon implementation.
\end{remark}

\subsection{Tip Kinematics}\label{sec:kin}

For a constant-curvature arc of length $L$,
\begin{equation}
p_x = \frac{\sin(\kappa L)}{\kappa},\qquad p_z = \frac{1-\cos(\kappa L)}{\kappa},
\label{eq:kin}
\end{equation}
with continuous limits $p_x\to L$, $p_z\to0$ as $\kappa\to0$. Let $J_\kappa(\kappa)=dp/d\kappa\in\mathbb{R}^{2\times1}$ be the translational Jacobian and let $n$ be the controlled tip-normal direction. The scalar normal coordinate is $y=n^\top p$, with
\begin{align}
J_n(\kappa)
&=\frac{dy}{d\kappa}=n^\top J_\kappa(\kappa),\\
\Lambda_n(\kappa)
&=\left(J_n(\kappa)M^{-1}(\kappa)J_n(\kappa)\right)^{-1}>0,
\label{eq:JnLambda}
\end{align}
away from configurations where $J_n=0$. This scalar operational-space inertia normalizes the effective tip impedance in the controllable direction.

\subsection{Contact Model and the Force-Safety Interaction Constraint}

During contact, a tip force $F_{\text{tip}}$ acts at the tip; in simulation the tissue is Kelvin--Voigt, $F_{\text{tissue}} = k_t\delta + b_t\dot\delta$, with penetration $\delta=\max(0,p_{\text{surf}}-p)$, tissue stiffness $k_t\approx 2$--20\,kN/m, and damping $b_t$. The contact maps to curvature through virtual work, $\tau_{\text{ext}}=J_\kappa^\top F_{\text{tip}}$; along the controlled normal this reduces to $\tau_{\text{ext}}=J_n F_n$, where $F_n=n^\top F_{\text{tip}}$. The clinically relevant perforation threshold $F_{\text{safe}}\approx 0.5$\,N~\cite{yokoyama2008} motivates a hard predicted normal-force bound, $|F_n(t)|\le F_{\text{safe}}$, with additional approach-velocity limits needed for environment-induced impact transients (Section~\ref{sec:verify}).

\subsection{Interaction-Dynamics Objective}

Given a reference $(p_d,\dot p_d,\ddot p_d)$, define the scalar interaction state $x=[e,\dot e]^\top$, where $e=y_d-y$ is the error in the controllable normal coordinate. This state characterizes the regulated catheter--tissue interaction dynamics: motion error, velocity error, persistent unknown loading, and the constraints that limit allowable contact. The objective is to choose $u(t)$ so $e(t)\to0$ in the nominal free-space limit while respecting tendon limits $u_{\text{lo}}\le u\le u_{\text{hi}}$, predicted normal-force safety $|F_n|\le F_{\text{safe}}$, and curvature limits $\kappa_{\text{lo}}\le\kappa\le\kappa_{\max}$ under bounded unknown contact, friction, and hysteresis.

\section{Predictive Interaction Dynamics Control}\label{sec:design}

The control architecture has two layers:
\begin{equation}
u = \underbrace{u_{\text{ff}}}_{\text{Layer 1: interaction-dynamics normalization}} + \underbrace{J_n F_{\text{mpc}}}_{\text{Layer 2: predictive correction}},
\label{eq:twolayer}
\end{equation}
with feedforward $u_{\text{ff}} = \hat C(\kappa,\dot\kappa)\dot\kappa + \hat K(\kappa) + J_n\Lambda_n(\kappa)\ddot y_d$ canceling the \emph{nominal} dynamics. In the single-DOF setting $F_{\text{mpc}}\in\mathbb{R}$ is the scalar corrective force along the controlled tip-normal; the transmitted generalized input is $J_nF_{\text{mpc}}$, dimensionally consistent with $u$. In the remainder we write $J_n$ for the scalar normal Jacobian and reserve $J_\kappa$ for the full translational Jacobian. For a rigid robot Layer~1 is exact; for a catheter $\hat C,\hat K$ are uncertain, so the cancellation is partial and the residual is absorbed by the disturbance estimate---more robust than inverting an unreliable model.

Since the system is single-DOF (Section~\ref{sec:model}), the error $e=y_d-y$ and the corrective input $F_{\text{mpc}}\in\mathbb{R}$ are scalar. With $x_e=[e,\dot e]^\top\in\mathbb{R}^2$,
\begin{equation}
\dot x_e = \underbrace{\begin{bmatrix}0&1\\0&0\end{bmatrix}}_{A_c\,(\text{const})} x_e + \underbrace{\begin{bmatrix}0\\-\Lambda_n^{-1}(\kappa)\end{bmatrix}}_{B_c(\kappa)} F_{\text{mpc}} + \underbrace{\begin{bmatrix}0\\1\end{bmatrix}}_{E_c\,(\text{const})} d(t),
\label{eq:errdyn}
\end{equation}
where $d$ lumps modeling error, contact, and the kinematic coupling $\dot J_n\dot\kappa$, normalized to acceleration units. This is the key configuration-invariant interaction-dynamics model: $A_c$ is \textbf{constant}, while only the scalar gain $B_c(\kappa)$ varies through $\Lambda_n(\kappa)$. The double integrator is the resulting normalized interaction model, not the main claim. The same constant-$A_d$, parameter-varying-$B_d$ structure exploited in the base framework~\cite{cao2026phri} and, earlier, in min--max MPC under input saturation~\cite{cao2005minmax}, enables offline precomputation of the prediction matrices. Because $A_c$ is nilpotent, the exact ZOH discretization is closed-form:
\begin{equation}
A_d=\begin{bmatrix}1&\Delta t\\0&1\end{bmatrix},\quad B_d(\kappa_k)=\begin{bmatrix}-\tfrac12\Lambda_n^{-1}(\kappa_k)\Delta t^2\\-\Lambda_n^{-1}(\kappa_k)\Delta t\end{bmatrix}.
\label{eq:disc}
\end{equation}

Augment with a scalar integrating disturbance $\hat d\in\mathbb{R}$:
\begin{equation}
\begin{bmatrix}x_{e,k+1}\\\hat d_{k+1}\end{bmatrix}=\begin{bmatrix}A_d&G_d\\0&1\end{bmatrix}\begin{bmatrix}x_{e,k}\\\hat d_k\end{bmatrix}+\begin{bmatrix}B_d(\kappa_k)\\0\end{bmatrix}F_{\text{mpc},k},
\label{eq:aug}
\end{equation}
with $G_d=[\tfrac12\Delta t^2,\Delta t]^\top$ the ZOH of $E_c$---\textbf{distinct from} $B_d$, which carries the $-\Lambda_n^{-1}$ input gain ($d$ is in acceleration units, so $G_d$ is not $\Lambda_n^{-1}$-scaled). A steady-state Kalman filter estimates $\hat d$ from tip-error measurements; the pair is observable with $C_{\text{aug}}=[I_2,0]$ since $G_d\neq 0$. No force sensor is required.

A predictive interaction-dynamics QP can be constructed from the augmented dynamics \eqref{eq:aug} and an input-centered cost. For a frozen or predicted inertia sequence, let
$U_d(\hat d)=[\Lambda_{n,0}\hat d,\ldots,\Lambda_{n,N-1}\hat d]^\top$ be the steady force sequence that cancels the estimated acceleration disturbance in \eqref{eq:errdyn}. The condensed QP is
\begin{equation}
\min_{U}\ \tfrac12 U^\top H\,U
+ \Big(\Gamma^\top\bar Q\big(\Phi x_{e,k}+\Delta(\hat d)\big)-\bar R U_d(\hat d)\Big)^\top U,
\label{eq:qp}
\end{equation}
with $U=[F_{\text{mpc},0};\ldots;F_{\text{mpc},N-1}]\in\mathbb{R}^N$, $H=\Gamma^\top\bar Q\Gamma+\bar R$, $Q=\text{diag}(K_d,D_d)$, and $\Phi\in\mathbb{R}^{2N\times2}$, $\Gamma$ precomputed once (constant $A_d$). When no interaction constraint is active the solution is the closed form
\begin{equation}
U^\star=-H^{-1}\!\left[\Gamma^\top\bar Q\big(\Phi x_e+\Delta(\hat d)\big)-\bar R U_d(\hat d)\right],
\label{eq:unconstrained_centered}
\end{equation}
a matrix--vector multiply. With constant $\Lambda_n$ over the horizon, $\Delta(\hat d)+\Gamma U_d(\hat d)=0$, so the change of variables $V=U-U_d$ reduces the optimizer to the nominal interaction-dynamics regulator in $V$. The reduced-model verification instantiates the corrective-force box with OSQP~\cite{osqp}; tendon and curvature rows can be added through the predicted actuation and state maps but are not exercised by the present MuJoCo benchmark.

\emph{Corrective-force constraint and tendon transmission.} The kinematic Jacobian $J_n=dy/d\kappa$ in \eqref{eq:JnLambda} maps curvature rate to tip velocity. The MuJoCo actuator map uses a different quantity, the dimensionless local tendon-to-tip-force transmission
\begin{equation}
\eta_T(\kappa):=\frac{\partial F_n}{\partial T},
\label{eq:tendon_transmission}
\end{equation}
measured from local press tests. The corrective input $F_{\text{mpc}}$ is expressed in tip-force units and realized as tendon tension $T_k=-(f_{\text{ff},k}+F_{\text{mpc},k})/\eta_T$ with tip-space elastic feedforward $f_{\text{ff}}=k_{\text{eff}}y_d$. The quasi-static controller-induced contact-force proxy is
\begin{equation}
\hat F_{n,k}=\eta_T\,|T_k|-k_{\text{eff}}\,|y_k|\approx k_{\text{eff}}e_k+F_{\text{mpc},k},
\label{eq:force_map}
\end{equation}
with any calibrated free-space friction bias represented by the residual state (Remark~\ref{rmk:compress}). In the benchmark, the elastic term is about 0.012\,N at the 1.41\,mm hold error, compared with the configured 0.5\,N threshold. The implemented pointwise limit is therefore
\begin{equation}
-F_{\text{safe}}\le F_{\text{mpc},k}\le F_{\text{safe}},
\label{eq:force_constraint}
\end{equation}
applied before mapping to the tendon command. This is a hard bound on the corrective-force proxy, not on the complete Kelvin--Voigt tissue reaction. The latter is measured in simulation and also depends on elastic, damping, alignment, and inter-sample effects; Section~\ref{sec:verify} quantifies the corresponding approach-velocity boundary.

\begin{theorem}[LQR--impedance equivalence]\label{thm:equiv}
In the unconstrained, disturbance-free limit ($d\equiv0$, no active constraints) the receding-horizon law is the static linear state feedback $F_{\text{mpc}}=K_{\text{eff}}\,e+D_{\text{eff}}\,\dot e$---the LQR-realized instance of the classical catheter impedance law $u_{\text{imp}}=\hat C\dot\kappa+\hat K+J_n(K_{\text{eff}}e+D_{\text{eff}}\dot e)$, realizing the effective tip-normal impedance $Z_{\text{eff}}(s,\kappa)\approx\Lambda_n(\kappa)s^2+D_{\text{eff}}s+K_{\text{eff}}$---where the \emph{realized} gains $(K_{\text{eff}},D_{\text{eff}})$ are the unconstrained LQR feedback for the weights $(Q,R)$. The unconstrained predictive interaction controller is therefore an LQR-tuned classical impedance, with the realized gains the Riccati image of the cost weights (Remark~\ref{rmk:realized}); it is \emph{not} a free rendering of a prescribed $(K_d,D_d)$ (following~\cite{cao2026phri}, Theorem~1).
\end{theorem}

\begin{IEEEproof}[Proof sketch]
Without constraints and with $d\equiv0$ the QP minimizer is the unconstrained stationary point $U^\star=-H^{-1}\Gamma^\top\bar Q\,\Phi x_e$, whose first block is a static gain $F_{\text{mpc}}=K_{\text{eff}}e+D_{\text{eff}}\dot e$; with the DARE terminal cost this is the infinite-horizon LQR feedback for $(Q,R)$, independent of $N$. After Layer-1 cancellation the residual plant is $\ddot e=-\Lambda_n^{-1}(\kappa)F_{\text{mpc}}+d$; multiplying by $\Lambda_n(\kappa)$ gives $\Lambda_n\ddot e+D_{\text{eff}}\dot e+K_{\text{eff}}e=\Lambda_n d$, the stated impedance, and adding the feedforward $u_{\text{ff}}=\hat C\dot\kappa+\hat K+J_n\Lambda_n\ddot y_d$ recovers $u_{\text{imp}}$. When constraints are active or $d\neq0$, the MPC departs from this static law---precisely its advantage.
\end{IEEEproof}

\begin{remark}[Prescribed vs.\ realized gains]\label{rmk:realized}
The cost weights $Q=\text{diag}(K_d,D_d)$, $R$ are design \emph{penalties}; the realized impedance $(K_{\text{eff}},D_{\text{eff}})$ is their LQR (Riccati) image, so in general $(K_{\text{eff}},D_{\text{eff}})\neq(K_d,D_d)$---in particular the cheap-control limit $R\to0$ yields a plant-determined gain, not $(K_d,D_d)$. To render a \emph{specific} $(K_d,D_d)$ exactly, prescribe the impedance gain directly ($F_{\text{mpc}}=K_d e+D_d\dot e$), making the equivalence exact at the cost of the predictive look-ahead (cf.~\cite{cao2026phri}, Remark~2).
\end{remark}

\section{Linear Closed-Loop Analysis}\label{sec:stability}

At a fixed nominal configuration, the DARE feedback associated with Theorem~\ref{thm:equiv} is Schur. Proposition~\ref{prop:lpv} quantifies how far the input gain may vary before that fixed linear feedback loses Schur stability. For a constant matched disturbance in the unconstrained nominal model, $\hat d\to d_\infty$ gives $U_d=\mathbf{1}_N\Lambda_n d_\infty$ and $\Delta(d_\infty)+\Gamma U_d=0$; the resulting equilibrium has zero tracking error. The MuJoCo benchmark evaluates the same mechanism using an integral force-bias state and reports the residual directly.

These statements concern the unconstrained linear equilibrium and fixed-gain LPV margin. A recursive-feasibility or constrained input-to-state-stability result would additionally require a specified terminal set, invariant feasibility argument, observer-error dynamics, and an online implementation of the stated horizon constraints. Those elements are reserved for the full horizon implementation and are not used to interpret the present simulation results.

\begin{proposition}[LPV stability margin]\label{prop:lpv}
For the scalar plant \eqref{eq:disc}, let $K=[k_1,k_2]$ be a gain designed at $\Lambda_{n,\text{ref}}$ (normalized to $\Lambda_{n,\text{ref}}=1$) and applied where the true inertia is $\Lambda_{n,\text{true}}$, giving the closed loop $A_{cl}(\rho)=A_d-\rho\,B_d(\Lambda_{n,\text{ref}})K$ with $\rho=\Lambda_{n,\text{ref}}/\Lambda_{n,\text{true}}$ (since $B_d\propto\Lambda_n^{-1}$). Then $A_{cl}(\rho)$ is Schur for all $\rho\in(0,\rho^\star)$ with
\begin{align}
\rho^\star
&=\min\!\left\{\frac{2}{\Delta t\,|k_2|},\
\frac{2}{\Delta t\,|k_2-\tfrac12\Delta t\,k_1|}\right\} \nonumber\\
&=\frac{2}{\Delta t\,|k_2|},
\label{eq:rhostar}
\end{align}
the first bound binding as a real pole exits through $z=-1$; the second corresponds to $\det=-1$. For the benchmark gain ($k_2=-294.9$, $\Delta t=2$\,ms) this gives $\rho^\star\approx3.39$, agreeing with the simulated spectral-radius sweep (Appendix~\ref{app:lpv}; \texttt{simulation/catheter\_verify.py}).
\end{proposition}

\begin{remark}[Workspace margin]\label{rmk:lpv}
Over the analyzed workspace $\Lambda_n(\kappa)$ varies $1.4\times$, i.e.\ $\rho\in[0.70,1.0]$, a margin of more than $3\times$ to $\rho^\star\approx3.4$ ($\Lambda_{n,\text{true}}>0.29\,\Lambda_{n,\text{ref}}$). The fixed gain remains Schur throughout this interval, and per-configuration $\Lambda_n(\kappa)$ scaling holds the closed-loop pole spread to $1.78\times10^{-6}$ in the reduced-model check (Section~\ref{sec:verify}). Under active force or tendon limits, the measured contact-limited residual replaces the unconstrained zero-error equilibrium.
\end{remark}

\section{Discussion}\label{sec:discussion}

\subsection{Predictive Interaction Optimization and the $\dot d=0$ Model}\label{sec:ddot}

The horizon formulation provides a structured route to trajectory feedforward and interaction constraints. The present benchmark isolates its static impedance specialization: known reference terms enter through feedforward, a pointwise limit bounds the corrective-force command, and an integral state rejects low-frequency residuals. A future online horizon implementation can replace the constant residual prediction with an ECG-gated periodic model, $\hat d(k{+}1)=\hat d_{\text{DC}}(k)+A_{\text{card}}\sin(\omega_{\text{heart}}(k{+}1)\Delta t+\phi)$, when force regulation rather than force limiting is the objective.

\subsection{Modeling and Sensing}

\emph{Partial cancellation.} Because catheter stiffness varies with hysteresis and specimen properties, the feedforward uses characterized elastic behavior and leaves the remaining mismatch to feedback and residual compensation. \emph{Sensing.} The simulated controller uses tip position and velocity without feeding tissue force back to the controller. Translation to hardware can use EM tracking or shape sensing, with direct force information included when available.

\subsection{Robustness to Contact-Normal Misalignment}

The scalar $y$ is taken along an assumed contact normal; proximal deflections can rotate the true normal by $\theta$. The effective contact stiffness along the assumed axis scales as $k_t\cos^2\theta$ (a $7\%$ change at $15^\circ$), while the true normal force associated with an axial cap scales conservatively as $1/\cos\theta$. In the MuJoCo sweep, the nominal 0.5\,N corrective-force cap permits a 0.539\,N realized estimate at $30^\circ$. Tightening the cap by $\cos30^\circ$ reduces the realized estimate to 0.458\,N and keeps the geometric upper bound at or below 0.5\,N for $\theta\le30^\circ$. This result motivates explicit normal-uncertainty margins in an integrated implementation.

\section{Interaction-Dynamics Simulation Results}\label{sec:sim}

\subsection{Setup}

We evaluate the controller on a \emph{distributed-compliance} MuJoCo plant: an eight-link pseudo-rigid-body tendon-driven catheter, total length 5\,cm, with torsional joint springs and a routed tension-only tendon. Tissue reaction is applied analytically at the tip using the Kelvin--Voigt law ($k_t=5$\,kN/m, $b_t=40$\,N$\cdot$s/m), while MuJoCo carries the distributed catheter dynamics. At the calibrated operating point, the measured tip-normal inertia is $\Lambda_n\approx3.5\times10^{-3}$, the local tendon transmission is $\eta_T\approx0.087$, and the catheter tip stiffness is $k_{\text{eff}}\approx8.4$\,N/m.

The executable benchmark uses the static impedance realization of Theorem~\ref{thm:equiv}, with gains scaled by the measured $\Lambda_n$, elastic feedforward $f_{\text{ff}}=k_{\text{eff}}y_d$, an integral residual state, and optional pointwise clipping of $F_{\text{mpc}}$. Tendon tension is $T=-(f_{\text{ff}}+F_{\text{mpc}})/\eta_T$ and is clipped to $[-8,0]$\,N. The benchmark measures $\Lambda_n$, $\eta_T$, and $k_{\text{eff}}$ once before each experiment; it does not run the horizon QP, Kalman filter, online gain scheduling, or curvature rows. The controller runs at 500\,Hz. The reference approaches the surface, presses to a target 1.5\,mm past it, holds, and retracts. The tissue stiffness, threshold, and penetrating reference are engineering stress-test parameters. The benchmark is reproduced by \texttt{simulation/catheter\_benchmark\_mujoco.py}.

\subsection{Controller Comparison}\label{sec:compare}

\begin{table}[!t]
\caption{Interaction-Dynamics Comparison Under Free Motion and Contact}
\label{tab:compare}
\centering
\footnotesize
\setlength{\tabcolsep}{3.5pt}
\begin{tabular}{lcccc}
\toprule
Controller & \shortstack{Approach\\RMS (mm)} & \shortstack{Max\\$F$ (N)} & \shortstack{$F_{\text{safe}}$\\viol.?} & \shortstack{Hold err.\\(mm)} \\
\midrule
Classical impedance & 0.28 & 0.17 & No & 1.48 \\
Impedance + integral & 0.03 & 0.60 & \textbf{Yes} & 1.41 \\
\textbf{Force-limited imp.+int.} & \textbf{0.03} & \textbf{0.47} & \textbf{No} & 1.41 \\
Joint-space PD & 1.82 & 0.00 & No & 3.76 \\
\bottomrule
\end{tabular}
\end{table}

\begin{figure}[!t]
\centering
\includegraphics[width=\columnwidth]{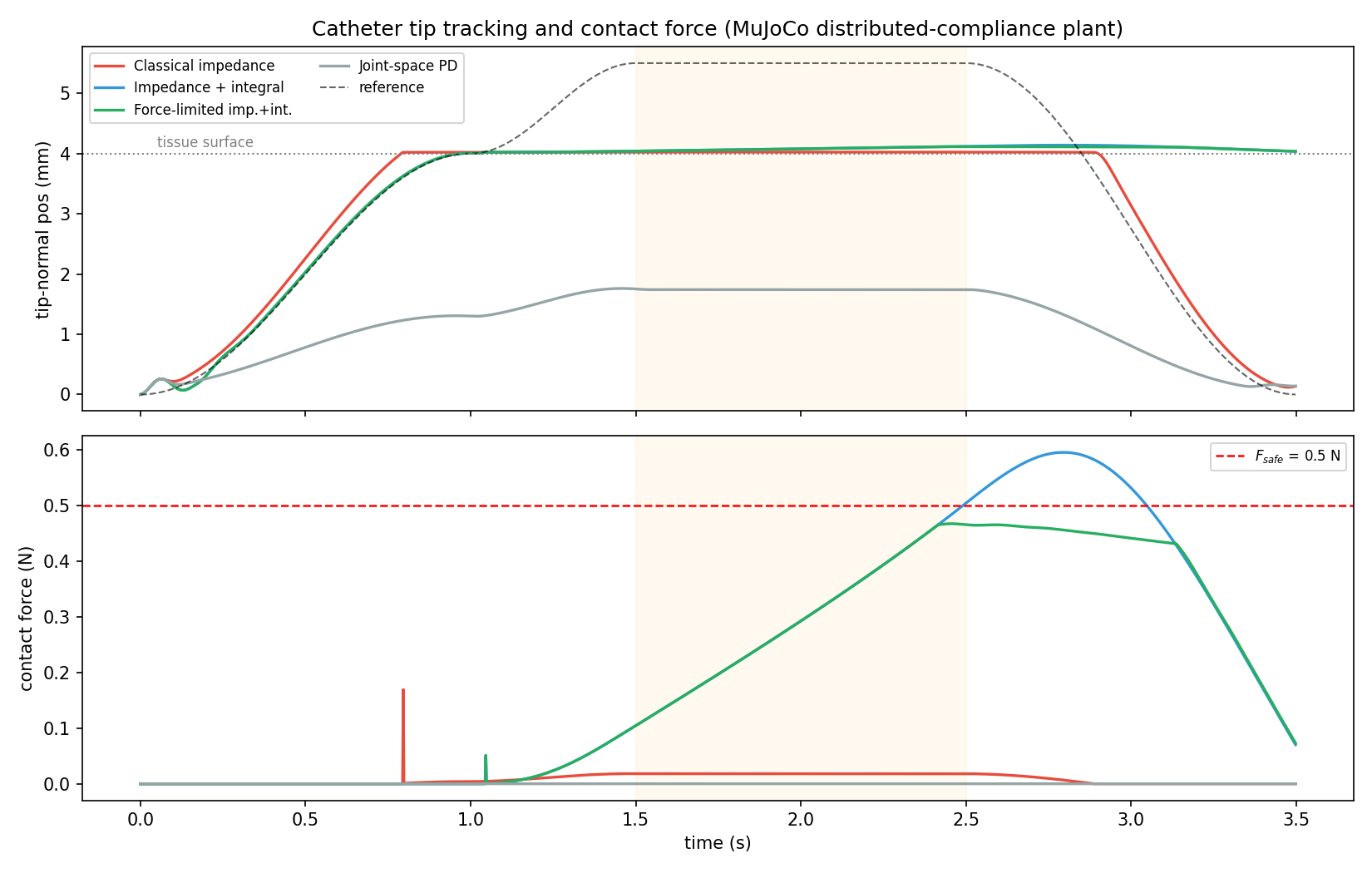}
\caption{Tip tracking (top) and measured Kelvin--Voigt tissue force (bottom) for the four controllers of Table~\ref{tab:compare}. The force-limited impedance-plus-integral controller (green) gives a 0.467\,N measured peak under the configured 0.5\,N threshold; removing the corrective-force limit (blue) raises the peak to 0.595\,N. Classical impedance reaches 0.169\,N with larger approach error, and joint-space PD does not reach the tissue in this task.}
\label{fig:benchmark}
\end{figure}

Key results (Table~\ref{tab:compare}, Fig.~\ref{fig:benchmark}):
\begin{itemize}
\item \emph{Free-space residual rejection.} Distributed bending and transmission mismatch leave 0.284\,mm approach RMS under classical impedance; the integral residual state reduces it to 0.028\,mm (90\%).
\item \emph{Corrective-force limiting.} With identical approach RMS, adding the pointwise $|F_{\text{mpc}}|\le0.5$\,N limit reduces measured peak tissue force from 0.595 to 0.467\,N. This isolates the value of explicit force-command limiting for the tested penetrating reference.\footnote{A supplementary video renders the unconstrained and force-constrained controllers side by side on the identical approach-press-hold-retract task: the unconstrained controller's contact force visibly crosses the 0.5\,N line while the constrained controller's does not, reproducing this 0.595 vs 0.467\,N result on the rendered episode.}
\item \emph{Contact-limited equilibrium.} The integral controllers retain about 1.41\,mm hold error because the 1.5\,mm penetrating target requires approximately 7.5\,N of tissue force while tendon saturation limits the realizable interaction. This is the expected constrained equilibrium, distinct from the zero-error equilibrium of the unconstrained nominal model.
\item \emph{Role of feedforward.} Joint-space PD gives 1.822\,mm approach RMS and does not reach the tissue, showing that the elastic feedforward is necessary for this gain normalization on the distributed plant.
\end{itemize}

\subsection{Comparison with Published Catheter-Control Approaches}

Published catheter controllers span impedance shaping, kinematic feedback, high-fidelity continuum models, disturbance observers, and force regulation~\cite{hogan1985,camarillo2008,yip2014,rucker2011,han2009}. Their platforms, sensing, contact models, objectives, and metrics differ, so categorical capability rankings would require a dedicated source-level systematic review. We therefore make the narrower comparison in Table~\ref{tab:quant}: published force-regulation results are shown beside the present simulation force-limit result to clarify the different objectives, not to claim head-to-head accuracy.

\begin{table}[!t]
\caption{Quantitative Comparison with Reported Sensor-Free / Force-Controlled Catheter Results}
\label{tab:quant}
\centering
\footnotesize
\setlength{\tabcolsep}{3pt}
\begin{tabular}{lcccc}
\toprule
Work & Objective & Sensing & Force result & Evidence \\
\midrule
\cite{jolaei2020} & force reg. & sensor-free & 0.03--0.05\,N RMSE & hardware \\
\cite{kesner2014} & force reg. & sensor+US & ${\sim}0.08$\,N RMSE & hardware \\
\shortstack[l]{\textbf{This work}\\(force limit)} & \shortstack[c]{pos.\ track\\+ cmd. limit} & sensor-free & \shortstack[c]{peak \textbf{0.47}\,N\\($\le0.5$ tested)} & simulation \\
\bottomrule
\end{tabular}
\end{table}

The rows address different objectives: \cite{jolaei2020,kesner2014} regulate contact force to a setpoint in hardware, whereas this study tracks position while limiting a corrective-force proxy and reporting the resulting measured tissue-force peak in simulation. Appendix~\ref{app:forcereg} separately characterizes an integral force-regulation mode. Accordingly, Table~\ref{tab:quant} is an objective-level context table; its numerical values should not be interpreted as a common-platform ranking.

\subsection{Residual Compensation}

The benchmark integral state reduces free-space approach RMS from 0.284 to 0.028\,mm. Under sign-dependent tendon friction, Section~\ref{sec:mujoco} shows a smaller but measurable reduction. The force-command limit acts independently of this residual state. Evaluation of the Kalman realization in \eqref{eq:aug}, including convergence and observer error, belongs to the online horizon implementation.

\subsection{Safety Under Cardiac Wall Motion}\label{sec:cardiac}

We rerun the force-limited impedance-plus-integral controller with a moving wall (Table~\ref{tab:cardiac}, Fig.~\ref{fig:cardiac}). Measured peak tissue force is 0.467\,N for the static wall, 0.468\,N at 0.3\,mm/1\,Hz, and 0.466\,N at 0.5\,mm/1.2\,Hz with 0.2\,mm position noise. The similarity follows from the modeled catheter tip stiffness ($k_{\text{eff}}\approx8$\,N/m), which is about 600 times lower than the tissue stiffness and allows the tip to follow the wall. These runs demonstrate robustness for the tested compliance and motion conditions; they do not constitute a motion-independent tissue-force certificate. Appendix~\ref{app:forcereg} treats the distinct force-setpoint objective.

\begin{table}[!t]
\caption{Force-Limited Impedance-plus-Integral Control Under Cardiac Wall Motion}
\label{tab:cardiac}
\centering
\footnotesize
\setlength{\tabcolsep}{4pt}
\begin{tabular}{lccc}
\toprule
Wall condition & \shortstack{Approach\\RMS (mm)} & \shortstack{Peak\\$F$ (N)} & \shortstack{$F_{\text{safe}}$\\viol.?} \\
\midrule
Static (Table~\ref{tab:compare} baseline) & 0.03 & 0.47 & No \\
$0.3$\,mm @ $1$\,Hz & 0.03 & 0.47 & No \\
$0.5$\,mm @ $1.2$\,Hz $+$ noise & 0.06 & 0.47 & No \\
\bottomrule
\end{tabular}
\end{table}

\begin{figure}[!t]
\centering
\includegraphics[width=\columnwidth]{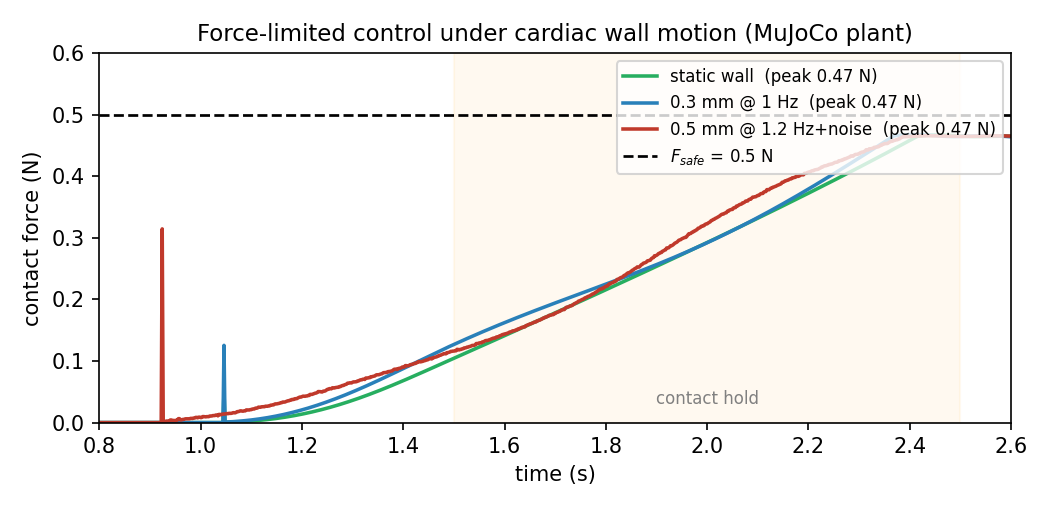}
\caption{Measured tissue force for the force-limited impedance-plus-integral controller under the three wall conditions of Table~\ref{tab:cardiac}. All tested peaks remain near 0.47\,N; the result is specific to the modeled catheter compliance, wall motion, and approach profile.}
\label{fig:cardiac}
\end{figure}

\subsection{Verification of Structural Claims}\label{sec:verify}

Four further checks---structural properties of the discrete-time controller and contact law, plus an approach-velocity qualification on the current 
MuJoCo plant---complete the picture (\texttt{simulation/catheter\_verify.py}):

\emph{(1) Configuration-adaptive compliance.} On the reduced constant-curvature model, recomputing the DARE gain at each $\Lambda_n(\kappa)$ over $\kappa\in[2,25]\,\text{m}^{-1}$ gives a dominant-pole spread of $1.78\times10^{-6}$ for the stated weights, compared with $1.50\times10^{-4}$ for a gain fixed at $\kappa=2\,\text{m}^{-1}$ (an 84-fold reduction). Figure~\ref{fig:poles} plots both loci. This is a per-configuration reduced-model design check; the headline MuJoCo benchmark measures $\Lambda_n$ once and does not schedule the gain online.

\begin{figure}[!t]
\centering
\includegraphics[width=\columnwidth]{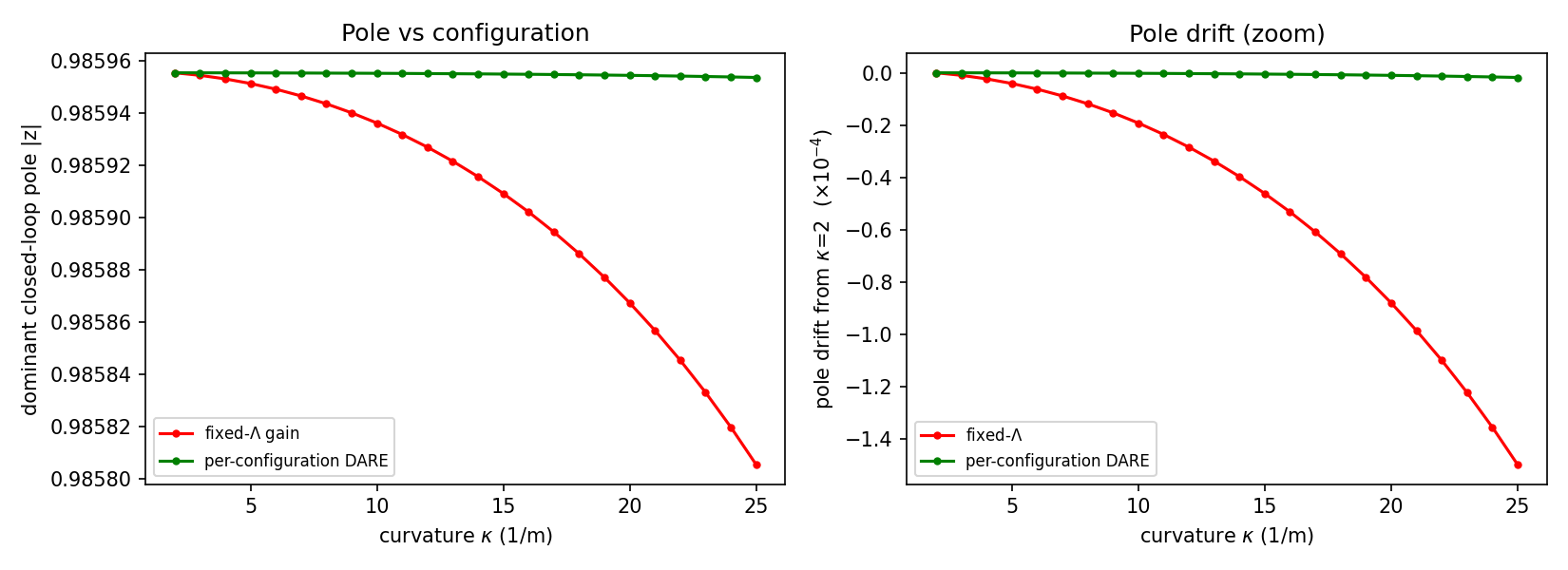}
\caption{Reduced-model dominant-pole magnitude and drift over curvature. Per-configuration DARE redesign (green) gives a $1.78\times10^{-6}$ spread for the stated weights, versus $1.50\times10^{-4}$ for the fixed nominal gain (red).}
\label{fig:poles}
\end{figure}

\emph{(2) Corrective-force box in the reduced-model QP.} For a 3\,mm error the unconstrained sequence reaches 1.9857\,N; OSQP with $|F_{\text{mpc},i}|\le0.5$\,N returns a first move of 0.5000\,N and respects the box over the horizon. With a 0.01\,mm error, the inactive-constraint solution agrees with the closed form to $1.77\times10^{-3}$ in infinity norm. This check covers the corrective-force box only; it does not exercise tendon, curvature, tissue-force, or estimator rows.

\emph{(3) Solver timing.} On an Intel Xeon w7-2475X under Windows 11 with Python 3.12.12 and OSQP 1.1.3, repeated runs measured 1.5--5.8\,$\mu$s per pre-inverted matrix--vector step and 32.4--37.4\,$\mu$s per warm-started OSQP update/solve, using 2000 repetitions per run. These solver-level measurements are below the 2\,ms sampling period on the test workstation; complete sensing, model-update, constraint-assembly, and actuator timing remains to be measured on the deployment target.

\emph{(4) Approach-velocity boundary.} In the constant-velocity MuJoCo sweep, measured peaks are 0.000, 0.220, 0.315, 0.364, 0.401, and 0.485\,N at 2, 4, 6, 8, 10, and 12\,mm/s, respectively. The peak rises to 0.591\,N at 15\,mm/s and 0.840\,N at 20\,mm/s. Thus the largest tested speed under the 0.5\,N threshold is 12\,mm/s, consistent with the damping guide $F_{\text{safe}}/b_t=12.5$\,mm/s. This measured transition shows why the corrective-force clip must be paired with velocity-limited approach or contact-onset reference shaping.

\subsection{Robustness to Unmodeled Tendon Hysteresis}\label{sec:mujoco}

The integral residual state is also evaluated with 0.3\,N Coulomb tendon friction on a free-space sinusoidal task. Without added friction, it reduces tracking RMS from 0.273 to 0.011\,mm (96\%); with sign-dependent friction, it reduces RMS from 0.806 to 0.468\,mm (42\%). The difference identifies motion reversal as the dominant residual not represented by a constant-bias state and motivates a hysteresis-aware internal model. Contact-force conclusions are drawn from the separate force-limited tests rather than from this free-space experiment.

\section{Conclusion}\label{sec:conclusion}

This work formulates steerable-catheter control as scalar tip-normal interaction-dynamics regulation. The reduced model yields a constant-$A_d$, parameter-varying-$B_d$ structure and a finite-horizon QP whose unconstrained, disturbance-free realization is classical impedance. The current distributed-compliance evaluation implements the corresponding static impedance law with elastic feedforward, integral residual compensation, tendon saturation, and a pointwise corrective-force limit.

On the eight-link MuJoCo plant, integral compensation reduces free-space approach RMS from 0.284 to 0.028\,mm. For the penetrating-reference task, the corrective-force limit reduces measured peak tissue force from 0.595 to 0.467\,N with unchanged approach accuracy. The 0.5\,N engineering threshold also holds for the tested cardiac-wall cases and for constant approach speeds through 12\,mm/s, while 15 and 20\,mm/s exceed it. These results characterize the force-limited static realization and its operating envelope; the separate reduced-model checks establish the DARE/LPV margin and corrective-force box behavior.

The next development stages are an online implementation of the horizon QP and Kalman estimator, explicit tendon/curvature/contact-output rows, hysteresis-aware and cardiac-phase internal models, and integrated-platform experiments with measured force and contact-normal uncertainty. Multi-segment extension will generalize the scalar formulation to independently actuated tendon coordinates and configuration-varying operational-space inertia.

\appendices

\section{Force-Regulation Characterization Under Cardiac Motion}\label{app:forcereg}

The main benchmark tracks position with a corrective-force limit, whereas~\cite{jolaei2020,kesner2014} regulate force to a setpoint. To characterize this second objective, we add a sensor-free force-regulation mode to the same MuJoCo plant. During free-space approach it tracks position; during contact it estimates the controller-induced force as $\hat F_c=\eta_T|T|-k_{\text{eff}}|y|$ and applies an integral force command toward an engineering target of 0.2\,N. The command is clipped to $[0,0.5]$\,N before tendon mapping. This supplementary experiment provides within-framework context rather than a common-platform comparison with the cited hardware studies. Two conditions are evaluated over the contact-hold window:
\begin{itemize}
\item \emph{Static wall:} force RMSE is 141.9\,mN about the 0.2\,N target, with a 0.347\,N peak. This is the baseline performance of the implemented integral force loop.
\item \emph{Moving wall:} with 0.3\,mm, 1\,Hz wall motion and 0.2\,mm position noise, force RMSE is 166.6\,mN and the peak is 0.434\,N. The increase motivates adding cardiac phase to the internal model when tight force-setpoint regulation is required.
\end{itemize}

Both measured peaks remain below the configured 0.5\,N threshold on this compliant plant. The result characterizes the implemented integral loop and does not change the main position-tracking comparison in Section~\ref{sec:compare}.

\section{Derivation of the LPV Stability Margin $\rho^\star$}\label{app:lpv}

This appendix proves Proposition~\ref{prop:lpv} in closed form, replacing the ``computing the spectral radius'' assertion with an explicit Jury 
analysis. After the exact ZOH discretization \eqref{eq:disc} the scalar error plant has the \emph{constant}
\begin{equation*}
A_d=\begin{bmatrix}1&\Delta t\\0&1\end{bmatrix},\qquad
B_d(\Lambda_n)=-\frac1{\Lambda_n}\,b,\quad b=\begin{bmatrix}\tfrac12\Delta t^2\\[1pt]\Delta t\end{bmatrix}.
\end{equation*}
A gain $K=[k_1,k_2]$ synthesized at $\Lambda_{n,\text{ref}}=1$ but acting where the true inertia is $\Lambda_{n,\text{true}}$ gives, 
using $B_d(\Lambda_{n,\text{true}})=(\Lambda_{n,\text{ref}}/\Lambda_{n,\text{true}})\,B_d(\Lambda_{n,\text{ref}})=\rho\,B_d(\Lambda_{n,\text{ref}})$ 
with $\rho=\Lambda_{n,\text{ref}}/\Lambda_{n,\text{true}}$,
\begin{align}
A_{cl}(\rho)
&=A_d-\rho\,B_d(1)\,K=A_d+\rho\,bK \nonumber\\
&=\begin{bmatrix}
1+\tfrac12\rho\Delta t^2k_1 &
\Delta t+\tfrac12\rho\Delta t^2k_2\\[2pt]
\rho\Delta t\,k_1 &
1+\rho\Delta t\,k_2
\end{bmatrix}.
\label{eq:Aclrho}
\end{align}
Its characteristic polynomial is $\chi(z)=z^2-\operatorname{tr}(\rho)\,z+\det(\rho)$ with
\begin{align}
\operatorname{tr}(\rho)&=2+\rho\Delta t\big(k_2+\tfrac12\Delta t\,k_1\big),\label{eq:tr}\\
\det(\rho)&=1+\rho\Delta t\big(k_2-\tfrac12\Delta t\,k_1\big).\label{eq:det}
\end{align}
(The $\rho^2$ terms in $\det$ cancel because $A_d$ is unipotent and $bK$ is rank one---the same constant-$A_d$ structure exploited throughout.) 
By Jury's criterion a monic second-order $\chi$ is Schur iff
\begin{align}
|\det(\rho)|&<1, \nonumber\\
\chi(1)&=1-\operatorname{tr}+\det>0, \nonumber\\
\chi(-1)&=1+\operatorname{tr}+\det>0.
\label{eq:jury}
\end{align}
Substituting \eqref{eq:tr}--\eqref{eq:det} into the three crossings:
\begin{itemize}
\item \emph{$z=+1$:} $\chi(1)=-\rho\Delta t^2 k_1>0$ for all $\rho>0$ (since $k_1<0$)---never binds.
\item \emph{$z=-1$:} $\chi(-1)=4+2\rho\Delta t\,k_2$, which vanishes at $\rho=-2/(\Delta t\,k_2)=2/(\Delta t\,|k_2|)$ (a real pole reaching $-1$).
\item \emph{$|\det|=1$:} $\det(\rho)<1$ holds for all $\rho>0$ (as $k_2-\tfrac12\Delta t\,k_1<0$), and $\det(\rho)=-1$ at $\rho=2/\big(\Delta t\,|k_2-\tfrac12\Delta t\,k_1|\big)$.
\end{itemize}
The stability boundary is the smallest positive crossing, giving \eqref{eq:rhostar}. For the benchmark gains $k_1=-2.04\times10^3$, $k_2=-294.9$ and $\Delta t=2$\,ms,
\begin{equation*}
\frac{2}{\Delta t\,|k_2|}=3.39,\qquad \frac{2}{\Delta t\,|k_2-\tfrac12\Delta t\,k_1|}=3.42,
\end{equation*}
so $\rho^\star=3.39$ (the $z=-1$ flip binds), exactly the value returned by the brute-force spectral-radius sweep over $\rho\in[1,40]$ in \texttt{catheter\_verify.py} ($\rho_{\max}=3.39$). The margin is governed by the velocity gain $k_2$ and the step $\Delta t$ alone; with the workspace $\rho\in[0.70,1.0]$ the loop is robustly stable with $>3\times$ margin (Remark~\ref{rmk:lpv}).

\balance

\end{document}